\documentclass[paper=a4, fontsize=11pt]{scrartcl} 
\usepackage[style=alphabetic-verb, maxbibnames=5,backend=bibtex]{biblatex}
\usepackage[utf8]{inputenc}
\usepackage[T1]{fontenc}
\usepackage{filecontents}

\usepackage[T1]{fontenc} 
\usepackage[english]{babel} 
\usepackage{amsmath,amsfonts,amsthm} 
\usepackage{algorithm}
\usepackage{float}
\usepackage[noend]{algpseudocode}
\usepackage{caption}
\usepackage{tikz}

\tikzset{mystyle1/.style={shape=circle,fill=black,scale=0.3}}
\tikzset{mystyle2/.style={shape=circle,fill=red,scale=0.45}}
\tikzset{mystyle3/.style={shape=circle,fill=blue,scale=0.37}}
\tikzset{mystyle4/.style={shape=circle,fill=green,scale=0.37}}

\tikzset{withtext/.style={fill=white}}

\usepackage{fancyhdr} 
\usepackage{mathtools}

\DeclarePairedDelimiter\floor{\lfloor}{\rfloor}

\makeatletter
\def\BState{\State\hskip-\ALG@thistlm}
\makeatother

\newtheorem{theorem}{Theorem}
\newtheorem{definition}[theorem]{Definition}

\newtheorem{lemma}[theorem]{Lemma}
\newtheorem{corollary}[theorem]{Corollary}
\newtheorem{proposition}[theorem]{Proposition}

\usepackage[titletoc,toc,title]{appendix}

\numberwithin{equation}{section} 
\numberwithin{figure}{section} 
\numberwithin{table}{section} 

\newcommand{\horrule}[1]{\rule{\linewidth}{#1}} 

\title{	
\normalfont \normalsize 
\horrule{0.5pt} \\[0.4cm] 
\huge Parallelepipeds obtaining HBL lower bounds \\ 
\horrule{2pt} \\[0.5cm] 
}

\author{James Demmel\footnote{EECS, Mathematics, University of California, Berkeley, CA (demmel@eecs.berkeley.edu)} \, Alex Rusciano\footnote{Mathematics, University of California, Berkeley, CA (rusciano@math.berkeley.edu)}} 

\begin{document}

\maketitle 

\begin{abstract}
\noindent This work studies the application of the discrete H\"{o}lder-Brascamp-Lieb (HBL) inequalities to the design of communication optimal algorithms.  In particular, it describes optimal tiling (blocking) strategies for nested loops that lack data dependencies and exhibit linear memory access patterns.  We attain known lower bounds for communication costs by unraveling the relationship between the HBL linear program, its dual, and tile selection.  The methods used are constructive and algorithmic.  The case when all arrays have one index is explored in depth, as a useful example in which a particularly efficient tiling can be determined.
\end{abstract}

\section{Background: H\"{o}lder-Brascamp-Lieb (HBL) Inequalities}
HBL inequalities are very powerful and include many famous inequalities, including H\"{o}lder's inequality and Young's inequality.  Stated for abelian groups $G=\mathbb{Z}^d, G_i=\mathbb{Z}^{d_i}$ with linear maps $\phi_i: G \rightarrow G_i$, they take the general form	

\begin{equation}
 \sum\limits_{x \in G} \prod\limits_{i\in J} f_i(\phi_i(x)) \leq \prod\limits_{j \in J} \|f_i\|_{1/s_i} \label{HBL}
\end{equation}

holding for non-negative integrable $f_i$ on $G_i$. The norms are $L_p$ norms.  The $s_i$ for which this holds depend on the maps $\phi_i$ only, and $G$ for us will be $\mathbb{Z}^d$.  We will call such $s$ feasible for the inequality (\ref{HBL}). It turns out the set of feasible $s$ form a polyhedron we will denote by $\mathcal{P}$.  The case when $G$ or the $G_i$ have a torsion component will not be discussed, although it is of potential interest.

\vspace*{.5cm}
As examples, the commonly stated version of H\"{o}lder's amounts to the inequality holding for $0 \leq s_1, s_2 \leq 1$ with $s_1 + s_2 = 1$ and the maps $\phi$ being the identity maps.  Young's inequality for convolutions uses maps from $\mathbb{Z}^2$ to $\mathbb{Z}$.  Also it should be noted that HBL inequalities were first introduced and proved for real or complex vector spaces, not abelian groups or rational vector spaces.

\vspace*{.3cm}
To connect this to communication avoidance, the group $G$ is the lattice $\mathbb{Z}^{d}$ of computations to perform.  To perform the computation corresponding to lattice point $x$, one needs to hold the data contained in $\phi_i(x)$ for all $i$.  The goal of HBL inequalities, in this context, is to bound how many lattice points $x$ we can compute using a memory size of $M$.  This means the $f_i$ are to be the indicator functions of some sets $S_i$.  Then HBL inequalities bound the size of set $S := \cap \phi_i^{-1}(S_i) \subset G$ that are computable provided we store the memory contained in each $S_i$.  

\vspace*{.5cm}
For any feasible $s$, and memory sizes $M_i = c_i \cdot M$ with $\sum c_i = 1$,

\begin{equation}
 |S| \leq \prod |S_i|^{s_i} = \prod_i M_i^{s_i} = M^{1^T s} \cdot \prod c_i^{s_i} \label{bound}
\end{equation}

\vspace*{.3cm}
Inequality (\ref{bound}) is a guaranteed upper bound given by the HBL inequality. As an approximation, minimizing the bound amounts to finding the smallest $1^T s$ over all feasible $s$ and neglecting the nuances behind the $c_i$.  We call this minimal sum $s_{\text{HBL}}$.  To formally state a rationale for ignoring the $\prod c_i^{s_i}$,

\begin{proposition}
 Assume $M$ total memory to work with as above and there are $n$ maps.  Then for the optimal choice of $c_i$, the HBL bound \ref{bound} is $\Theta(M^{1^T s})$.  Moreover, the generic choice $c_i = 1/n$ attains this bound.
 
\end{proposition}

Here and in the remainder of the paper, all big O notation is with respect to the memory parameter $M$, not the dimension of the computation lattice or number of maps.

\begin{proof}

 Recall the bound in inequality (\ref{bound}) is $M^{1^T s} \cdot \prod c_i^{s_i}$.  For an upper bound, take the $c_i \leq 1$; no $c$ can do better than this.  Now consider $c_i = 1/n$; perhaps there are better $c_i$ for the particular $s_i$, but we can at least do this well.  Consequently, the HBL bound for any $s$ is always in the range
 \[
  [\frac{1}{n^{1^T s}}M^{1^T s} , M^{1^T s}]
 \]
 The lower and upper ranges of this interval are $\Theta(M^{1^T s})$.
\end{proof}

\vspace*{.5cm}
Correspondingly, we will distinguish between two senses of optimality for obtaining the lower bounds.  

\begin{definition}
The family of sets $S(M)$, parametrized by integer $M$, form asymptotically optimal tilings if translations of the set $S(M)$ can tile $\mathbb{Z}^d$ and $S(M)$ satisfies 
\begin{equation}
|S(M)| = \Theta( M^{s_{\text{HBL}}}) \label{big}
\end{equation}
as well as
\begin{equation}
 \forall i, \, |\phi_i(S(M))| = O(M) \label{small}
\end{equation}
\end{definition}
For exact optimality, while $1^T s$ determines the asymptotic behavior of inequality (\ref{bound}), the inequality differs by a constant factor for every $c$ one chooses.  Choice of $c$ could be regarded as strategic use of memory.  This leads us to consider the sharpness of the following inequality when defining exact optimality:

\[
|S| \leq M^{s_{\text{HBL}}} \min_{\substack{s \, \in \, \mathcal{P} \\ 1^T s = s_{\text{HBL}}}}  \max\limits_{\sum c_i=1}  \prod\limits_i c_i^{s_i}
\]

Simple calculations give the optimal choice of $c_i$.  It is equivalent to maximize 
 
 \[
\sum\limits_i s_i \log(c_i)
\]
instead. The method of Lagrange multipliers implies the objective gradient and constraint gradients are parallel:
 \[
  (1, \dots, 1) = \lambda (\frac{s_1}{c_1}, \dots, \frac{s_n}{c_n})
 \]
The solution to this is $c_i = \frac{s_i}{1^T s}$.  This leads to a definition of exact optimality:
\begin{definition}
 Define scaling parameter
 \[
  \gamma := \frac{1}{(s_{\text{HBL}})^{s_{\text{HBL}}}} \min_{\substack{s \, \in \, \mathcal{P} \\ 1^T s = s_{\text{HBL}}}} \prod\limits_i s_i^{s_i} 
 \]
The family of sets $S(M)$, parametrized by integer $M$, are exactly optimal tilings if translations of $S(M)$ can tile $\mathbb{Z}^d$ and $S(M)$ satisfies
 \begin{equation}
  |S(M)| = (1-o(1)) \cdot \gamma M^{s_{\text{HBL}}}\label{exact1}
  \end{equation}
  as well as
  \begin{equation}
  \forall i, \, \sum\limits_i |\phi_i(S(M))| \leq M \label{exact2}
 \end{equation}
\end{definition}

The $o(1)$ term is required for any reasonable goal because one cannot allocate room for fractions of entries from arrays. Conceptually, we are requiring the ratio of $|S(M)|$ and the theoretical optimum to tend to one, i.e. the relative difference is going to 0.  

\vspace*{.5cm}
As another comment on the definition, the minimization problem of computing $\gamma$ is not difficult; elementary calculus shows that $\log(\prod s_i^{s_i})$ is convex.  Indeed, this is the negative entropy function.  Minimization of this function can even be done efficiently.

\vspace*{.5cm}
Most of this work focuses on asymptotic optimality, but we will discuss exact optimality in two special cases.
\section{Background: Solving for $s_{HBL}$ through an LP}
For the maps $\phi$, there is a natural necessary condition for the $s$:

\begin{proposition}
If $s_i$ are to satisfy inequality (\ref{HBL}), then it is necessary that for any subgroup $H$ of $G$ we have the following
 \[
  \sum s_i \cdot \text{rank}(\phi_i(H)) \geq \text{rank}(H)
 \]
\end{proposition}
 \begin{proof}
 To prove this fact, it suffices to use indicator functions; consequently we show the necessity in the case of Inequality (\ref{bound}).  
 
 \vspace*{.5cm}
 A uniformly growing cube $S(r)$ in the subgroup $H$ parametrized by side length $r$ grows like $r^{\text{rank}(H)}$ in volume because it is a $\text{rank}(H)$ dimensional object. This is the LHS of Inequality (\ref{bound}).
 
 \vspace*{.5cm}
The RHS of Inequality (\ref{bound}) needs to match this growth rate.  The images $\phi_i(S(r))$ grow asymptotically like $r^{\text{rank}(\phi_i(H))}$.   The LHS must be less than the RHS as $r\rightarrow \infty$, implying the result.
 
 \vspace*{.5cm}
 The above is more of a sketch; in any case, the above appears in \cite{CDKSY15} as part of Theorem 1.4.
 \end{proof}

 The surprising part of Theorem 1.4 is that this is actually sufficient.  \cite{CDKSY15}, supplemented by work from \cite{BCCT}, established
 \begin{theorem}\label{fpt}
 Given maps $\phi$, a collection $s_i \geq 0$ satisfies inequality (\ref{HBL}) if and only if they satisfy for all subgroups $H$ of $G$,
 \[
  \sum\limits_i s_i \cdot \text{rank}(\phi_i(H)) \geq \text{rank}(H)
 \]
\end{theorem}

Because there are only finitely many possible values of the rank of $H$ and its images, there exists a finite list of subgroups which is sufficient to generate the constraints.  Consequently the problem can be formulated as a linear program (LP), if we can find a sufficient list of subgroups.

\begin{definition}[HBL Primal LP] \label{primal}
 If $\mathbf{E}$ is a finite sufficient list of subgroups needed to give the correct HBL constant, we define the HBL primal LP to be

\[
  \begin{matrix} 
   
\text{minimize} &&& 1^T s \\

  \text{subject to} &&& \forall H \in \mathbf{E}, \, \sum s_i \cdot \text{rank}(\phi_i(H)) \geq \text{rank}(H) \\
  &&& s_i \geq 0
  
  \end{matrix}
\]
\end{definition}

\vspace*{.5cm}
Recent work has started to get a grip on formulating the LP in a computationally feasible manner.  In \cite{CDKSY15} a few things are established.  For one thing, only the lattice of subgroups generated by $\text{ker}(\phi_i)$ under sums and intersections needs to be used to generate inequality constraints in Theorem \ref{fpt}.  However, this lattice is often infinite in higher dimensions.  Also, \cite{CDKSY15} describes a terminating algorithm which discovers all the constraints needed to formulate an equivalent LP.  However, the algorithm's complexity is unknown.  The results of \cite{GGOW} provide a number of novel insights into algorithmic computation of $\mathcal{P}$ arising from the closely related continuous version of the inequalities, including a polynomial time membership and weak separation oracles. Much remains to be understood, and in general formulating and optimizing over $\mathcal{P}$ remains intractable.  Here is a summary of tractable special cases, as far as the authors are aware:

\begin{itemize}
 \item All maps are coordinate projections \cite{CDKSY13}.
 \item Each $\text{ker}(\phi_i)$ is rank 1,2, d-1, or d-2, mixes are allowed \cite{V}.  Stated for continuous version of inequalities.
 \item There are no more than 3 maps; then the kernel subgroup lattice is bound by 28, a classical result \cite{D}.
\end{itemize}

\section{Attaining the lower bound by duality}

For applications, it is just as interesting to attain lower bounds as to show they exist.  In this and the subsequent section, we show how the dual of the HBL Primal LP leads to an asymptotically optimal parallelpiped tiling of the computation lattice.  

\vspace*{.5cm}
Notate by $\mathbf{E} = (E_1, \dots, E_k)$ a finite list of subgroups used to formulate the HBL Primal LP.  We introduce the notation

\[
\text{rank}(\mathbf{E}) := (\text{rank}(E_1), \dots, \text{rank}(E_k))^T, 
\]
and similarly 
\[
\text{rank}(\phi_i(\mathbf{E})) := (\text{rank}(\phi_i(E_1)), \dots, \text{rank}(\phi_i(E_k)))^T
\]

One may write the dual of the HBL Primal LP as

\[
  \begin{matrix} 
   
\text{maximize} &&& y^T \text{rank}(\mathbf{E}) \\

  \text{subject to} &&& \forall \phi_i, \, \,  y^T \text{rank}(\phi_i(\mathbf{E})) \leq 1 \\
  &&& y_i \geq 0
  
  \end{matrix}
\]
  
 It will be useful to be more flexible in how we think of the dual problem.  The dual, as formulated from the primal, comes with a particular subgroup list $\mathbf{E}$ indexing the dual variables.  The methods we use add and remove subgroups from consideration, and we do not wish this to fundamentally change the dual LP.  In the future we use the notation $\mathbf{E}$ for the analogous, but more generic, role in the revised dual:
  
\begin{definition}[Dual LP]\label{dual}
   Recall the HBL setting consists of maps $\phi_i$ from lattice $\mathbb{Z}^d$.  A dual vector $y$ will be considered to be indexed by all subgroups of $\mathbb{Z}^d$, but with finitely many non-zero coordinates.  The non-zero coordinates are defined to be the support of $y$.  If the list of subgroups $\mathbf{E}=(E_1, \dots, E_t)$ is the support of $y$, then we introduce a few notations and definitions.  Two natural notational shorthands  are 
   \[
    y^T \text{rank}(\mathbf{E}) := \sum y_{E_j} \text{rank}(E_j)
   \]
   and 
   \[
   y^T \text{rank} (\phi_i(\mathbf{E})) := \sum y_{E_j} \text{rank}(\phi_i(E_j))
   \]
   Now define the objective value of $y$ to be 
   \begin{equation}
   \text{val}(y) := y^T \text{rank}(\mathbf{E}) \label{objective}
   \end{equation}
   and say $y$ is feasible if it satisfies the conditions
   
   \begin{gather}
    \forall \phi_i,  \, C_i(y) := y^T \text{rank}(\phi_i(\mathbf{E})) \leq 1 \label{feasibility1}\\
    y_i \geq 0 \label{feasibility2}
   \end{gather}

\end{definition}

As a further notational note on this definition, we use a few symbols in place of $\mathbf{E}$ when extra information is present.  Typically we will use $\mathbf{Y}$ when the supporting subgroups are independent, and $\mathbf{U}$ when they are a flag.  These definitions are covered later.
 
  \vspace*{.5cm}
  This is readily interpretable when the supporting subgroups $\mathbf{Y}$ of the dual vector are independent.  Here independent means that $\text{rank}(\oplus_i Y_i ) = \sum_i \text{rank}(Y_i)$.  Before explaining exactly how we interpret the dual, we need the following geometrically intuitive lemma.  It demonstrates that asymptotic optimality eases some difficulties stemming from discreteness.

\begin{lemma}\label{GCL}
  
  Take any independent elements $e_1, \dots, e_h $ contained in rank $h$ subgroup $Y \subset \mathbb{Z}^d$ and linear mapping $L$. Define the set
   
   \begin{equation}
   S := \{z \in \mathbb{Z}^d | z = \sum a_i e_i \, \text{with} \, \, 0 \leq a_i \leq \floor{M^k} - 1, \, a_i \in \mathbb{Z} \} \label{cube}
   \end{equation}
   
   In this equation, $k$ is an arbitrary positive number, and $M$ is an integer conceptually representing memory capacity.  In applications later, $k \in (0, 1]$.
   
   \vspace*{.5cm}
   Then for any linear map $L$, $|S| = \Theta(M^{kh})$ and $|L(S)| = O(M^{kr})$ where $r := \text{rank}(L(Y))$.  In applications later, $L$ is taken to be one of the $\phi_i$.
  \end{lemma}
  
  \begin{proof}
   The elements in set $S$ are $O(M^k)$ from the origin in Euclidean distance, hiding the dimensional factor $d$ in the big O notation.  By linearity, the elements of $L(S)$ are also $O(M^k)$ from the origin in $\text{im}(L)$. Therefore an $r$ dimensional cube residing within $L(Y)$ with side lengths $O(M^k)$ can contain $L(S)$.  This means that $|L(S)| = O(M^{kr})$.
   
   \vspace*{.5cm}
   Finally, from independence of the $e_i$ it follows that $|S| = {\floor{M^k}}^h$
  \end{proof}

We now define the parallelepiped-like construction that will be used to create asymptotically optimal tilings.  Although not the only possible way to build good tiling shapes, it is flexible and leads to clean descriptions.

\begin{definition}[Product Parallelepiped]\label{PP}
     Suppose we are given a dual vector $y$, whose non-zero values are attached to a set of independent subgroups $Y_1, \dots, Y_t$.  Form $S_{Y_i}$ as in Eq. \ref{cube}, using $k = y_{Y_i}$.  
   
   \vspace*{.5cm}
   Now define a parallelepiped shaped tile through a Minkowski sum of sets 
   
   \begin{equation}
   S := S_{Y_1} + \dots + S_{Y_t} \label{sum}
   \end{equation}
\end{definition}

The independent elements used to construct $S_{Y_i}$ are left unspecified; the choice affects constants, but will not affect asymptotic optimality.

\vspace*{.5cm}
To make things explicit, Algorithm \ref{tiling} below produces the translations needed to tile $\mathbb{Z}^d$ with set $S$.  It uses an important matrix factorization for linear maps between abelian groups (or more generally between modules over a principle ideal domain) known as the Smith Normal Form.  

\vspace*{.5cm}
The Smith Normal Form of a matrix $A$ with integer entries is of the form $A = UDV^{-1}$ where $U,V$ are unimodular and $D$ is diagonal with non-negative integer entries. Its diagonal entries $d_i := D_{ii}$ are uniquely defined by requiring $d_i | d_{i+1}$.  In our use of the factorization, the matrix $A$ is injective and consequently full rank, so that all $d_i$ are non-zero.

\begin{algorithm} [H]
\begin{algorithmic}[1]
\caption{Construct tile $S$ and its translations $T$ that tile $\mathbb{Z}^d$}
\label{tiling}
\State Input: Memory parameter $M$ 
\State Input: For each $i=1, \dots, t$: independent elements $e_{i1}, \dots, e_{i h_i}$ chosen from independent rank $h_i$ subgroups $Y_i$
\State Input: For each $i=1, \dots, t$: memory scaling parameter $k = y_{Y_i}$ for subgroup $Y_i$
\State Output: translations $T$ of the set $S$ that tile $\mathbb{Z}^d$ 
\State $E \leftarrow (e_{11}, e_{12}, \dots, e_{t h_t})$
\State $S \leftarrow \{ E \cdot (a_{11},a_{12},\dots, a_{t h_t})^T | \, a_{ij} \in \{0, \dots, \floor{M^{y_{Y_i}}}-1 \}$
\State $m \leftarrow \sum h_i$
\State $(U, D, V) \leftarrow \text{Smith Normal Form}(E)$
\State $U' \leftarrow \text{last d-m columns of U}$
\State $U'' \leftarrow \text{first m columns of U}$
\State $T_1 \leftarrow \{E \cdot (a_{11}, a_{12}, \dots, a_{t h_t})^T \, | \, a_{ij} \in \floor{M^{y_{Y_i}}} \cdot \mathbb{Z}\}$
\State $T_2 \leftarrow \{ U' \cdot (a_1, \dots, a_{d-m})^T \, | \, a_i \in \mathbb{Z}\}$
\State $T_3 \leftarrow \{ U'' \cdot (b_1, \dots, b_m)^T \, | \, b_i \in \{0, \dots, d_i - 1\} \}$
\State $T \leftarrow \text{ Minkowski sum } T_1  + T_2 + T_3$

\Return $S,T$
\end{algorithmic}
\end{algorithm}

The set $S$ returned by the algorithm exactly follows Def. \ref{PP}.  The translations $T$ come from two sources: $T_1$ accounts for the finite size of $M$ while tiling the subgroup generated by $e_{11}, e_{12}, \dots, e_{t h_t}$ under integer linear combinations.  In the future we will write this subgroup as $\langle e_{11}, e_{12}, \dots e_{t h_t} \rangle$, and similarly for other generated subgroups. The others $T_2$, $T_3$ account for the need to tile each coset in $\mathbb{Z}^d / \langle e_{11}, e_{12}, \dots e_{t h_t} \rangle$.

\begin{proposition}
  Algorithm \ref{tiling} correctly outputs a parallelepiped set $S$ which under translation by $T$ tiles $\mathbb{Z}^d$.  This holds for any input: that is, for any selection of independent subgroups $U_i$, choice of independent elements within these subgroups, memory parameter setting $M$, and memory scalings $M^{y_{Y_i}}$.
\end{proposition}

\begin{proof}

Let us begin by discussing the translations in $T_1$.  Consider $x =  \sum a_{ij} e_{ij}$.  The set $S$ only uses scalings of $0$ to $\floor{M^{y_{Y_i}}}-1$ of each $e_{ij}$. Consequently, for $x$ to be contained in a translation of $S$, the shift component in the $e_{ij}$ direction must be in the range 
\[
[a_{ij}-\floor{M^{y_{Y_i}}}+1, a_{ij}]
\] 
The only such member of $T_1$ has $e_{ij}$ component $\floor{a_{ij}/{\floor{M^{y_{Y_i}}}}}\cdot \floor{M^{y_{Y_i}}}$

\vspace*{.5cm}
$T$ needs to also contain exactly one representative of each coset in 
\[
\mathbb{Z} / \langle e_{11}, e_{12}, \dots e_{t h_t} \rangle \simeq \mathbb{Z}^{d-m} \oplus (\bigoplus_{i=d-m}^{d} \mathbb{Z} / D_i \mathbb{Z}) 
\]

$T_2$ accounts for the free component, and $T_3$ for the torsion component.

\vspace*{.5cm}
Indeed, let $U'a + U''b, U'a' + U''b'$ be two distinct elements of $T_2 + T_3$.  Saying they are in the same coset is exactly saying their difference lies in $\text{im}(E)$.  Writing $E = U D (V)^{-1}$ as in Algorithm \ref{tiling} and noting $V$ is unimodular, it is clear that $\text{im}(E) = \text{im}(U D)$.  Conclude that lying in the same coset is equivalent to
\[
 U'a+U''b - U'a'-U''b' \in \text{im}(U D) 
\]
As $U$ is unimodular, This means for some $c \in \mathbb{Z}^d$
\[
 (b,a)^T-(b',a')^T = D c
\]
Because $D$ is $d$-by-$m$, the the last $d-m$ coordinates of $D c$ are 0.  This means $a = a'$.  Also for $b, b'$ to be used in $T_3$, they must satisfy $0 \leq b_i, b_i' < d_i$.  But then
\[ 
 -d_i < b_i - b'_i = d_i c_i < d_i
\]
which is only possible if $b_i = b_i'$.

\vspace*{.5cm}
To conclude that all cosets are represented, we argue that for any $x \in \mathbb{Z}^d$, there are $U'a,U''b$ such that $x - U (b,a)^T \in \text{im}(E)$.  Take  $b_i = (U^{-1}x)_{i} \, \text{ mod } (D_i)$ and $a_i = (U^{-1}x)_{m+i}$.

\end{proof}

\vspace*{.5cm}
This paper includes examples at the end in Appendix \ref{examples}.  These will help demonstrate this approach and future aspects of the paper.

\vspace*{.5cm}
Now that the fundamental tiling object and mechanism have been described, we now begin to analyze the properties of the tile $S$ in relation to the HBL problem.

  \begin{proposition}\label{ISP}
  
   Suppose we are given a dual vector $y$, whose non-zero values are attached to a list of independent subgroups $\mathbf{Y} = (Y_1, \dots, Y_t)$.  Form the product parallelepiped $S$ of Def. \ref{PP}.  
   
   \vspace*{.5cm}
   Then $|S|$ = $\Theta(M^{y^T \text{rank}(\mathbf{Y})})$.  If in addition $y$ is dual feasible, then $|\phi_i(S)| = O(M)$ holds for each $\phi_j$.
   \end{proposition}

\begin{proof}

\vspace*{.5cm}
 By independence of the subgroups contained in $\mathbf{Y}$, it follows that $|S| = \prod\limits_i |S_{Y_i}|$.   Apply the count estimates of Lemma \ref{GCL} to this:
 \[
  |S| = \prod\limits_i \Theta(M^{\text{rank}(Y_i) \cdot  y_{Y_i}}) = \Theta(M^{y^T \text{rank}(\mathbf{Y})})
 \]
It remains to consider the images of this set under the $\phi_j$ in the case $y$ is feasible.  This requires a bound on $|\phi_i(S)|$  Invoking the count estimates of Lemma \ref{GCL} in the second inequality, and feasibility property (\ref{feasibility1}) of $y$ in the third inequality
 
 \[
 |\phi_j(S)| \leq \prod\limits_i |\phi_j(S_{Y_i})| \leq \prod\limits_i M^{\text{rank}(\phi_j(Y_i))y_{Y_i}} = M^{C_j(y)} \leq M
 \]

\end{proof}

It will be necessary to strengthen the bounds on the $|\phi_i(S)|$ later, but this already proves a useful result:

\begin{corollary}
 Suppose there exists a dual optimal solution $y$ with non-zero dual variables attached to a set of independent subgroups $\mathbf{Y}$.  Then we may tile the lattice $\mathbb{Z}^d$ with an asymptotically optimal parallelpiped shape.
\end{corollary}

\begin{proof}
 From strong duality, $y^T \text{rank}(\mathbf{Y}) = s_{HBL}$.  Then Proposition \ref{ISP} provides a parallelepiped shape that is asymptotically optimal, as it satisfies equations \ref{big} and \ref{small}.  Algorithm \ref{tiling} shows how to tile using object $S$.
\end{proof}

This construction breaks down when dual variables correspond to non-independent subspaces.  The next section generalizes the results of this section through a notion of flags of the subgroup lattice.

\section{Attainability in General}

In this section, we describe an algorithm for producing an asymptotically optimal tiling.  The recipe is to formulate the primal, solve the corresponding dual, and iteratively modify the solution of the dual to something geometrically interpretable.  Consequently, at least asymptotically, the HBL lower bounds are attainable by a polyhedral tiling, and to do so is essentially no harder than describing the set of feasible $s$ for inequality (\ref{HBL}).  This set is commonly referred to as the Brascamp-Lieb Polyhedron.

\subsection{Flags}

It might not be possible to find a dual vector supported on independent subgroups that obtains the optimal value.  However, it turns out that it is possible to find one supported on what we here define to be a flag.

\begin{definition}
 A flag of the lattice $\mathbb{Z}^d$ (for us) is a sequence $\mathbf{U}$ of strictly nested subgroups
 \[
 \emptyset \subset U_1 \subset \dots\subset U_t = \mathbb{Z}^d
 \]
 
\end{definition}

 We want to take the dual solution, and transform it to being supported on a flag.  The following is a simple but important property in accomplishing this goal.  It was also helpful in =\cite{CDKSY15} and \cite{V}, the latter of whom we note found flags useful in studying the vertices of the Brascamp-Lieb polyhedron.

\begin{lemma}[Substitution Lemma]\label{SL}
 For any linear map $L$ on $\mathbb{Z}^d$ and subgroups $V,W$, 
 \[
\text{rank}(L(V)) \geq \text{rank}(L(V \cap W)) + \text{rank}(L(V + W)) - \text{rank}(L(W))
\] 
On the other hand, 
\[
\text{rank}(V) = \text{rank}(V \cap W) + \text{rank}(V + W) - \text{rank}(W) 
\]

\end{lemma}
\begin{proof}
The claimed equality in the lemma follows by writing a basis for $V \cap W$ and completing it to a basis for $W$ with a second set of independent basis elements.  Call the subgroup spanned by the second set $P$.  Observe $P$ has trivial intersection with $V$, and the rank of $P$ is $\text{rank}(W)-\text{rank}(V \cap W)$.  Applying these observations to $W+V = P+V$,
\[
 \text{rank}(W+V)=\text{rank}(P+V) = \text{rank}(P)+\text{rank}(V) = \text{rank}(W)-\text{rank}(W \cap V) + \text{rank}(V)
\]

establishing the result.  To prove the inequality, apply the equality to subspaces $L(V)$, $L(W)$, and then observe
\[
L(V \cap W) \subseteq L(V) \cap L(W)  , \,  \text{while} \,  L(V + W) = L(V)+L(W)
\]
The reason for the possible inequality is that maybe there are different elements $v \in V$ and $w \in W$, but $L(v) = L(W)$.
\end{proof}

We employ this observation repeatedly to shift the support of a dual vector onto a flag, through the following procedure.  It takes as input a feasible $y$ supported on an arbitrary list $\mathbf{E}$ and outputs a feasible $y'$ supported on a flag $\mathbf{U}$ with the same objective value Eq. \ref{objective}.  Recall by feasible we mean Eq. \ref{feasibility1}, \ref{feasibility2} are satisfied.

\begin{algorithm} [H]
\begin{algorithmic}[1]
\caption{Find dual feasible vector supported on a flag}
\label{MD}
\State Input: dual feasible vector $y$ supported on $E_1, \dots, E_m$
\State Output: feasible $y'$ supported on a flag $U_1, \dots U_t$, with the same objective value as $y$
\State Initialize $y'$ as $y$
\While{$y'$ is not supported on a flag}
\State $V,W \leftarrow $ any pair in the support of $y'$ NOT satisfying $V \subset W$ or $W \subset V$
\State Let $V$ be the member of the pair with $y_V' \leq y_W'$
\State $y_W' \leftarrow y_{W}' - y_V'$
\State $y_{V + W}' \leftarrow  y_{V + W}' + y_V'$
\State $y_{V \cap W}' \leftarrow y_{V \cap W}' + y_V'$ (if $V \cap W \neq \{0\}$)
\State $y_V' \leftarrow 0$
\EndWhile
\Return $y'$ and its support
\end{algorithmic}
\end{algorithm}

\begin{theorem}[Non-negative Flag Theorem]\label{NFT}
 Algorithm \ref{MD} is correct: given input a dual feasible vector $y$ supported on $E_1, \dots, E_m$, it outputs a dual feasible $y'$ supported on a flag $\mathbf{U} = (U_1, \dots U_t)$ with the same objective value as $y$.
\end{theorem}

\begin{proof}
 The existence of the pair $V,W$ is equivalent to the support of $y'$ not being totally ordered, which is equivalent to the support of $y'$ not being a flag.  So if the algorithm does terminate, the support will be a flag.  We must show the algorithm terminates, and that $y'$ maintains the objective value and feasibility.
 
 \vspace*{.5cm}
 Induction establishes that $y'$ is always non-negative.  Indeed, inside the while loop, the only danger is $y_W' - y_V'$.  But $y_V'$ is the smaller of the two by construction.  So $y'$ always satisfies Eq. \ref{feasibility2}.
 
 \vspace*{.5cm}
 Let $y''$ denote the value of $y'$ after another pass through the while loop.  We examine the effect of the iteration on Eq. \ref{feasibility1}, \ref{objective}.  In the case $V \cap W \neq \{0\}$,
 \[
  C_i(y'') = C_i(y') - y_{V}' \left[\text{rank}(\phi_i(W)) - \text{rank}(\phi_i(V \cap W))-\text{rank}(\phi_i(V + W)) +\text{rank}(\phi_i(V))\right]
 \]
 The bracketed quantity is non-negative by Lemma \ref{SL}, meaning Eq. \ref{feasibility1} still holds.  If $V \cap W = \{0\}$, then 
 \begin{align*}
C_i(y'') &= C_i(y') - y_{V}' \left[\text{rank}(\phi_i(W)) -\text{rank}(\phi_i(V + W)) +\text{rank}(\phi_i(V))\right] \\
 &= C_i(y') - y_{V}' \left[\text{rank}(\phi_i(W)) - \text{rank}(\phi_i(V \cap W))-\text{rank}(\phi_i(V + W)) +\text{rank}(\phi_i(V))\right]
\end{align*} 
 where we used $\text{rank}(\phi_i(V \cap W)) = 0$.  Consequently Lemma \ref{SL} applies again.  Similarly, the objective value is preserved: in the case of $W \cap V \neq \{0\}$,
 \[
 \text{val}(y'') = \text{val}(y') -  y_{V}' \left[\text{rank}(W) - \text{rank}(V \cap W)-\text{rank}(V+ W) +\text{rank}(V)\right]
 \]
with the bracketed quantity being $0$ by Lemma \ref{SL}.  As before, the same follows in the case $V \cap W = \{0\}$ by noting $\text{rank}(V\cap W) = 0$.
 
 \vspace*{.5cm}
 It remains to establish that the algorithm will terminate.  At first glance, it appears that the $y'$ might cycle in the algorithm.  However, each iteration is increasing the dual variables on $V+W$ and $V \cap W$, so the dual vector seems to be shifting towards the high and low rank subgroups.
 
 \vspace*{.5cm}
 To capture this intuition, we define a simple measure of extremeness on dual vectors.  Recall all groups reside in $\mathbb{Z}^d$.  To a dual vector $y$ we assign a list $w(y)$ of length $d$.  To do this, set
 
 \[
  w(y)_i = \sum\limits_{U \in \text{ support}(y), \text{ rank}(U)=i} y_{U}
 \]
 
 For example, if $y$ is supported on $\langle e_1, e_2 \rangle , \langle e_1 \rangle, \langle e_2 \rangle $ with values $1, .5, 2$, and $d= 3$, then $w(y) = (2.5, 1, 0)$.
We say $y'$ is \textit{more extreme} than $y''$ if $w(y')$ is reverse lexicographically more than $w(y'')$.  Every iteration of the while loop makes $y'$ more extreme; indeed, the value $y_{V+W}$ increases and $V+W$ is of strictly larger rank than $V$ or $W$.
 
 \vspace*{.5cm}
 Now we show that $w(y')$ can take on only finitely many values, completing the proof.  Observe that $y'^T 1$ stays the same or decreases each iteration, so coordinates of $w(y)$ are bound by $y'^T 1$.  Also, the values produced by the algorithm come from performing only addition and subtraction operations on the the coordinates of $y$, which are rational.  Consequently coordinates of $w(y')$ lie in the finite set
 \[
\text{span}_{\mathbb{Z}}(y_{E_1}, \dots, y_{E_m}) \cap [0, y^T 1]
\]
 
\end{proof}

\subsection{Parallelepiped Tilings from Flags}

The main theorem of the previous section allows us to transform an optimal dual vector into another optimal dual vector supported on a flag.  Now we convert the flag subgroups into independent subgroups in the natural manner in order to produce a tiling shape.

\begin{definition}[Flag Parallelepiped]\label{FP}
 Suppose $y$ is supported on flag $\mathbf{U}$.  Let $\mathbf{Y}$ be a sequence of independent subgroups such that $Y_1+\dots+Y_i = U_i$. Define the dual vector $y'$ supported on $\mathbf{Y}$ by 
 \[
 y_{Y_i}' =  y_{U_i} + \dots + y_{U_t} 
 \]
Form a product parallelepiped $S$ of Def. \ref{PP} from $y'$.  We will call $S$ the flag parallelepiped of $y$, and $y'$ its associated dual vector.
\end{definition}

Here let's briefly summarize the progress so far, and what we still need to accomplish.  Provided we formulated the HBL Primal LP and solved its dual, we found a feasible $y$ with objective value $s_{HBL}$.  From Theorem \ref{NFT}, this $y$ can be modified to another $y'$ supported on some flag, maintaining the objective value $s_{HBL}$ and feasibility as defined in Eq. \ref{objective}, \ref{feasibility1}, \ref{feasibility2}.  Next apply the flag parallelepiped construction of Def. \ref{FP} to $y'$ to create a tile $S$ and its associated $y''$.  Proposition \ref{ISP} implies that $S$ includes $\Theta(M^{s_{HBL}})$ lattice points.  However, $y''$ might no longer satisfy Eq. \ref{feasibility1}, so Proposition \ref{ISP} does not show that $|\phi_i(S)| = O(M)$.  We need to expand the analysis of Lemma \ref{GCL} to the case of parallelepipeds instead of cubes.

\begin{lemma}[Growing Parallelepiped Lemma]\label{GPL}
 Consider independent subspaces $Y_1, \dots, Y_t$ with corresponding dual values $y_{Y_i}$.  Construct the product parallelepiped as in Def. \ref{PP} from these independent spaces and dual values.  Assume the subgroups are ordered so that $y_{Y_i}$ monotonically decreases with $i$.  In keeping with Def. \ref{FP} have $U_i := Y_1 + \dots + Y_i$ for $i=1, \dots , t$, and for convenience $U_0 := \{0\}$.  For any linear map $L$, set 
 \[
d_i := \text{rank}(L(U_i))-\text{rank}(L(U_{i-1}))
\]
 
Then we have the bound
 
\[
|L(S)| = O \left( \prod M^{y_{Y_i} \cdot d_i}\right)
\]
In particular, this holds for $L$ chosen to be any of the $\phi_j$.
\end{lemma}

Before beginning the proof, we remark on the significance.  The weaker bound used in Proposition \ref{ISP} was
\[
 |L(S)| \leq \prod |L(S_{Y_i})| = O \left(\prod M^{y_{Y_i}\cdot a_i} \right)
\]
with $a_i := \text{rank}(L(Y_i))$.  From independence of the subgroups $Y_j$, it is immediate that $d_i \leq a_i$.  For example, when $L$ is the identity, $d_i = a_i$.  However, when $L(Y_i)$ is not independent of $L(U_{i-1})$, it is always the case that $d_i < a_i$.

\begin{proof}

\vspace*{.5cm}
The goal is to propose a rectangular prism $T$ containing $L(S)$.  Of the defining edges, $d_i$ of them will be length $O(M^{y_{Y_i}})$.  This would prove the needed bound. 

\vspace*{.5cm}
Intuitively, we just need to make the $d_1$ dimensions coming from $L(Y_1)$ have the largest size $O(M^{y_{Y_1}})$, and the next $d_2$ dimensions coming from $Y_2$ will need to have length $O(M^{y_{Y_2}})$ and so forth.  To formally show this by constructing $T$, it is convenient to interpret all subgroups instead as subspaces of $\mathbb{Q}^d$ with the standard Euclidean inner product and its induced norm.  Now apply a Gram-Schmidt orthogonalization procedure to the sequence $L(Y_1), L(Y_2), \dots, L(Y_t)$.  This yields subspaces $E_1, \dots E_t$ satisfying

\[
 E_1 = L(Y_1), \, E_1 +\dots + E_i = L(Y_1) + \dots + L(Y_i), \, E_i \perp E_j \, \text{for} \, i \neq j
\]

\vspace*{.5cm}
Take $T$ to be the Minkowski sum formed by cubes $T_i$ of side length $O(M^{y_{Y_i}})$ growing in the spaces $E_i$.  It is readily observed that $|T| = O\left(\prod M^{y_{Y_i}\cdot d_i}\right)$.  Denote by $P_{E_i}$ the orthogonal projection onto $E_i$.  If we can show $P_{E_i}(L(S)) \subset T_i$ for each $i$, then $L(S) \subset T$.  The proof would then be complete.

\vspace*{.5cm}
Select an arbitrary $x \in S$. That is,

\[
L(x) = L(x_{Y_1}) + \dots + L(x_{Y_t})
\]
where $x_{Y_j} \in S_{Y_j}$.  Observe that $L(x_{Y_i}) \in \text{ker} (P_{E_j})$ for $i < j$.  Also $\|L(x_{Y_j})\|_2 = O(M^{y_{Y_j}})$.  This implies
\[
 P_{E_i}(L(x)) = P_{E_i}(L(x_{Y_i})) + \dots + P_{E_i}(L(x_{Y_t}))
 \]
 and therefore
 \[
 \|P_{E_i}(L(x))\|_2 = O(M^{y_{Y_i}}) + \dots + O(M^{y_{Y_t}}) = O(M^{y_{Y_i}})
\]

As $T$ is permitted to be $O(M^{y_{Y_i}})$ in $E_i = \text{im} (P_i)$, we conclude that $P_i(S) \subset T$ if the hidden constant for $T$ large enough.
\end{proof}

This readily applies to the construction of Def. \ref{FP}:
\begin{theorem}[Non-negative Parallelepiped Theorem]\label{NPT}
From an optimal dual feasible vector $y$ supported on a flag $\mathbf{U}$, form a flag parallelepiped $S$.  Then $|\phi_j(S)| = O(M)$ for each $\phi_j$ in the HBL problem, and $|S| = \Theta(M^{s_{\text{HBL}}})$.
\end{theorem}

\begin{proof}
Let $y'$ be the associated dual vector of $S$ with independent subgroups $Y_1, \dots, Y_t$ as described in Def. \ref{FP}.  As $y$ has positive entries, $y_{Y_i}'$ are monotonically decreasing.  Consequently, Lemma \ref{GPL} applies.  It implies that 

  \begin{align*}
    |\phi_j(S) | = O(\prod\limits_i M^{y'_{Y_i} \cdot d_i}) &= O(\prod\limits_i M^{(y_{Y_i}+\dots + y_{Y_t}) \cdot d_i}) \\
    &= O(M^{\sum_i y_{Y_i}\cdot (d_1+\dots+d_i)}) &= O(M^{y^T \text{rank}(\phi_j(U)) }) = O(M)
\end{align*}

That $|S| = \Theta(M^{s_{\text{HBL}}})$ follows from Proposition \ref{ISP}
\end{proof}

This is the major theoretical result.  Combined with earlier results, it notably establishes the central claim:
\begin{corollary}
If one is able to produce a sufficient list of subgroups for the HBL primal, then one can determine an asymptotically optimal tiling shape. 
\end{corollary}
\begin{proof}
 Solve the dual LP to get $y$.  Apply Theorem \ref{NFT} to produce $y'$ optimal and supported on a flag.  Then apply the flag parallelepiped construction to $y'$ yielding tile $S$, which is asymptotically optimal by Theorem \ref{NPT}.
\end{proof}

 \section{Two Cases of Exactly Optimal Tilings}
The preceding sections have focused on attaining asymptotic optimality.  We would like the tiling shape to meet the requirements of exact optimality, as given by equations \ref{exact1} and \ref{exact2}.  This would matter for communication avoiding applications in practice.  However, we do not have a characterization of when this is possible.  Instead, we will describe two simpler situations in which we can provide exactly optimal tilings.
 
\subsection{Rank One Maps}
We begin by discussing the case of rank 1 maps.  This could be regarded as a generalized n-body problem.  Detailed work on the communication patterns and bounds for the n-body problem was examined in \cite{DGKSY} and \cite{KY}.  This corresponds to arrays with single indices.  First, we demonstrate what $s_{HBL}$ is for this case and a method for obtaining asymptotic optimality.

\begin{proposition}
 Assume the maps $\phi_i$ are rank 1 with $i \in J$ and $|J| = n$, and the lattice is $\mathbb{Z}^d$.  If $\cap_i \, \text{ker}(\phi_i) = \emptyset$, then $s_{HBL}=d$.  Then a $d$ dimensional cube with sides $O(M)$ is asymptotically optimal.  Otherwise, the Primal LP of Def. \ref{primal} is infeasible.
\end{proposition}

\begin{proof}
First, suppose that $\cap_i \, \text{ker}(\phi_i)$ is nonempty.  Then take $E_1$ to be a non-zero element of this intersection; as kernels are subspaces, the $\langle E_1 \rangle$ is also in the kernel.  This implies the corresponding constraint in Def. \ref{primal} is
\[
 0^T s \geq 1
\]

which can't be satisfied. So the LP is infeasible, meaning one could get ``infinite'' data re-use.  See the appendix for an example.

\vspace*{.5cm}
Now suppose $\cap \text{ker}(\phi_i) = \emptyset$.  One subgroup you could use is $\mathbb{Z}^d$ itself.  By the rank 1 assumption, the inequality constraint in Def. \ref{primal} corresponding to this subgroup is

\[
 1^T s \geq d
\]

\vspace*{.5cm}
This implies $s_{HBL} \geq d$.  Now we exhibit a feasible primal vector $s$ for which $1^T s = d$ to complete the proof.  One may select a subset $J' \subset J$ with $|J'| = d$ such that $\cap_{i \in J'} \text{ker}(\phi_i) = \emptyset$.  This follows by induction; start with $H_0 = \mathbb{Z}^d$.  Then recurse by $H_i = H_{i-1}\cap \text{ker}(\phi_i)$.  If $\text{rank}(H_i) = \text{rank}(H_{i-1})-1$ then include $i$ in $J'$.  Because belonging to $\text{ker}(\phi_i)$ amounts to satisfying a single linear equation, the rank may only decrease by $1$.  Choose the primal variable $s$ to be $1_{J'}$.  

\vspace*{.5cm}
It remains to establish the feasibility of this $s$.  The argument may proceed recursively as above.  This time label the elements of $J'$ to be $i_1, \dots , i_d$, and let $T$ denote any subgroup.  Set $H_{i_0} = T$ and $H_{i_j} = H_{i_{j-1}} \cap \text{ker}(\phi_i)$ recursively.  Again, ranks of the $H_{i_j}$ decrease by $1$ or stay the same, compared to the rank of $H_{i_{j-1}}$.  The end result is $\{0\}$; this implies $T$ is not a strict subset of at least $\text{rank}(T)$ of the kernels associated with $J'$.  Consequently $s$ satisfies the constraint of Def. \ref{primal} for the subgroup $T$:

\[
 s^T \text{rank}(\phi(T)) \geq \text{rank}(T)
\]

This implies the feasibility of $s$ and establishes $s_{HBL} = d$.  Observe the dual variable indicating the space $\mathbb{Z}^d$ achieves the value $s_{HBL}$ as well.  By Lemma \ref{GCL}, a cube with sides $O(M)$ is asymptotically optimal.
\end{proof}

This establishes that the running Algorithm \ref{tiling} on $\mathbb{Z}^d$ produces an asymptotically optimal tiling.  For exact optimality, we must restrict to the case where there are $d$ rank-one maps with empty kernel intersection.

\begin{lemma}[Basis Lemma]\label{BL}
The subgroup $\cap_{j \neq i} \text{ker}(\phi_j)$ is rank 1; take $e_i$ to be a non-zero element of smallest Euclidean norm from this subgroup.   Then each subgroup $\text{ker}(\phi_i)$ contains the independent elements $e_1, \dots, e_{i-1}, e_{i+1}, \dots, e_d$.
\end{lemma}

\begin{proof}
We must check the $e_i$ are well defined, and that they are linearly independent.

\vspace*{.5cm}
Every time we intersect with one of the kernels, the rank reduces by 1.  The empty intersection property of the $d$ kernels implies this; every intersection adds a linear constraint, and if one of the linear constraints turned out to be redundant then $d$ intersections would not result in the empty set.

\vspace*{.5cm}
Lastly, we make sure that the $e_i$ are independent.  If not, then some $e_i$ is in the span of the other $e_{i'}$; however, the other $e_{i'}$ are contained in $\text{ker}(\phi_{i})$.  This means $e_i \in \text{ker}(\phi_i)$ is as well.  Then $e_i$ lies in the intersection of all the kernels, which by assumption is empty set.
\end{proof}

This basis is critical in the following; 
\begin{proposition} \label{exact}
 Let $e_i$ be as in Lemma \ref{BL}.  Then the sets $S := \{ \sum a_i e_i  | a_i \in \mathbb{Z}, \, 0 \leq a_i \leq \floor{M/d} -1 \}$ are exactly optimal.  That is, the output of Algorithm 1 on independent elements $e_1, \dots, e_d$ of $\mathbb{Z}^d$ meets the requirements of Eq. \ref{exact1} and \ref{exact2}.
\end{proposition}

\begin{proof}
 The first part of this section established that the optimal $s_{HBL}$ is $d$ and comes from each $s_i = 1$. This is in fact the unique solution to the primal LP of Def. \ref{primal} so it is by default the minimizer of $\gamma$ in Eq. \ref{exact1}.  Alternatively, evenly distributed values $s_i$ maximize entropy and consequently would minimize $\gamma$.  Plugging this in, $c_i = \frac{1}{d}$ and $\gamma = \frac{1}{d^d}$.
 
 \vspace*{.5cm}
 It remains to confirm that $|S| = (M/d)^d + O(1)$ and $\sum_i |\phi_i(S)| \leq M$.  First, by independence of the $e_i$, there are $\floor{M/d}^d$ lattice points enclosed.  Now if $M = a\cdot d + r$,
 \[
  (M/d)^d = a^d \cdot (1 + \frac{r}{M})^d = \floor{M/d}^d \cdot (1 + \frac{r}{M})^d \leq \floor{M/d}^d e^{r/M} = \floor{M/d}^d (1 + o(1))
 \]

 This establishes Eq. \ref{exact1}.  For the memory bound constraint, consider $\phi_i(S)$.  Applied to any point $z \in S$, it outputs $a_i \cdot e_i$.  As $a_i$ only varies between $\floor{M/d}$ values, the result follows.
 
\end{proof}

We may summarize the approach to tiling in Proposition \ref{exact} in the the following algorithm.

\begin{algorithm} [H]
\begin{algorithmic}[1]
\caption{Exactly Optimal Tiling, Rank One Maps}
\label{rankone}
\State Input: rank one maps $\{\phi_i\}_{i=1}^d$ with coordinate representations $a_i \in \mathbb{Z}^d$, satisfying $\cap \text{ker}(\phi_i) = \{0\}$, memory parameter $M$
\State Output: tile $S$ and translations by $T$ that tile $\mathbb{Z}^d$
\State Initialize $e_1, \dots, e_d \in \mathbb{Z}^d$
\For{i = 1 to d}
\State $A \leftarrow (a_1, \dots, a_{i-1}, a_{i+1}, \dots, a_d)^T$
\State $U, \, D, \, V \leftarrow \text{Smith Normal Form}(A)$
\State $e_i \leftarrow  \text{column 1 of }V$
\EndFor
\State $S, \, T \leftarrow$ Algorithm \ref{tiling} on input subgroup $U_1 = \mathbb{Z}^d$, its independent elements $e_1, \dots, e_d$, memory parameter $M$, and scaling $y_{\mathbb{Z}^d}=1$

\Return $S$ and $T$
\end{algorithmic}
\end{algorithm}

\vspace*{.5cm}

The new component of the algorithm is calculating the independent elements $e_i$. With this in mind, we examine the calculation of the $e_i$.  Recall $e_i \in \cap_{j \neq i} \text{ker}(\phi_j)$ of smallest Euclidean norm are used in Proposition \ref{exact}.  The implies $e_i$ is in the kernel of matrix $A_i := (a_1,\dots, a_{i-1}, a_{i+1}, \dots, a_d)^T$.  Decompose this matrix by Smith Normal Form, giving the representation $U D V^{-1}$.  Because the rank of $A_i$ is $d-1$, only the first diagonal entry of $D$ is $0$.  This means the kernel of the matrix is exactly what $V^{-1}$ maps to $(1, 0,\dots, 0)^T$, meaning multiples of the first column of $V$.  As $V$ is unimodular, this column is also the shortest integer valued multiple of itself.

\subsection{Rank d-1 Maps}
This section follows the rank 1 case very closely, and consequently is kept brief.  As an example, this setting includes the case of matrix multiplication and therefore much of linear algebra.  We again discuss asymptotic optimality, followed by exact optimality.

\begin{proposition}
 In the case where all maps are rank $d-1$, the optimal dual vector $y$ can be taken to have the subspace generated by the kernels of all the maps as its only nonzero coordinate. Call this subspace $H$ and let $k= \text{rank}(H)$.  Then $y_W = 1/(k-1)$ is optimal for the dual LP of Def. \ref{dual}. In addition, $s_{HBL} = k/(k-1)$.
\end{proposition}

\begin{proof}
 As $H$ is sent to a rank $k-1$ space by each of the $\phi_i$, $y$ is indeed dual feasible with objective value $k/(k-1)$.
 
 \vspace*{.5cm}
 We must show that that this matches the HBL lower bound, as then by strong duality $y$ is dual optimal.  Propose the primal value $s= 1/(k-1) \cdot 1_A$, where $1_A$ indicates any $k$ maps whose (one-dimensional) kernels generate the rank $k$ space.  Essentially, this is saying the kernels of these maps are independent.
 \vspace*{.5cm}
 
 Then consider any rank $l$ subgroup $T$, and its images under the maps in $A$.  By independence of kernels in the construction of $A$, only $l$ of the maps might send this to a rank $l-1$ group, the others send it to a $l$ dimensional space.  As there are $k$ non-zero $s_i$, the LHS of the constraint given by subgroup $T$ in the primal LP of Def. \ref{primal} is
 \begin{align*}
\text{rank}(\phi(T))^Ts &= \sum\limits_{\phi_i \in A} s_i \cdot \text{rank}(\phi_i(T)) \\ 
&= \frac{1}{k-1} \sum\limits_{\phi_i \in A} \text{rank}(\phi_i(T)) \\
&= \frac{1}{k-1} \left[ (l-1) \cdot \#\{\phi \in A | \text{ker}(\phi) \cap T \neq \{0\}\} + l \cdot \#\{\phi \in A | \text{ker}(\phi) \cap T = \{0\}\} \right]   \\
 &\geq \frac{1}{k-1} \left[ (l-1)\cdot l + l \cdot (k-l) \right] \\
 &= (l^2-l+lk-l^2)/(k-1)= l\cdot(k-1)/(k-1)=l
\end{align*}
Meanwhile, the RHS is $l$, so the constraint is satisfied.
\end{proof}

Similar to the rank 1 case, for exact optimality, restrict to when the kernels of the $\phi_i$ are independent.  Again let $e_i$ denote a non-zero smallest Euclidean norm representative of $\text{ker}(\phi_i)$, and let $E$ be the subgroup they span.
\begin{proposition}
Suppose the number of maps is equal to $k$ and the kernels are independent. Form the set $S := \{ \sum a_i \cdot e_i  | a_i \in \mathbb{Z}, 0 \leq a_i \leq \floor{\frac{M}{k}}^{\frac{1}{k-1}}-1 \}$.  That is, apply Algorithm \ref{tiling} to the independent elements $e_i$ of subgroup $E$, with scaling $y_E = 1/(k-1)$ and memory parameter $M/k$.  Then $S$ meets the criteria of Eq. \ref{exact1} and \ref{exact2} for exact optimality.
\end{proposition}
\begin{proof} (Sketch). 
 We established that $s_i = 1/(k-1)$ has $1^T s = s_{\text{HBL}}$.  Moreover, because the values are evenly distributed, it minimizes $\gamma$.  Plugging this in, $c_i = 1/k$, $\gamma = (\frac{1}{k})^{k/(k-1)}$.
 
 \vspace*{.5cm}
 The remainder follows analagously the argument of rank one maps: show that $M/d$ being rounded induces $1 +o(1)$ relative difference between $\gamma \cdot M^{s_{\text{HBL}}}$ and $|S|$ for Eq. \ref{exact1} to be satisfied, and then quickly confirm Eq. \ref{exact2} holds.
 
 \vspace*{.5cm}
\end{proof}

\section{Acknowledgments}

This research is supported in part by NSF grants ACI-1339676, 
and DOE grants DOE CS-SC0010200, DOE DE-SC008700, DOE AC02-05CH11231, 
ASPIRE Lab industrial sponsors and affiliates Intel, Google, Hewlett-Packard,
Huawei, LGE, NVIDIA, Oracle, and Samsung. Other industrial supporters are Mathworks and Cray.
Any opinions, findings,
conclusions, or recommendations in this paper are solely those of the
authors and does not necessarily reflect the position or the
policy of the sponsors.

\newpage

\printbibliography

\newpage
\begin{appendices}
\section{Examples}\label{examples}
We will use this appendix to concretely demonstrate certain key points and techniques in the paper.  Each example emphasizes different aspects.
\subsection{Rank-one maps}
\begin{figure}
\centering
    \begin{minipage}{.5\textwidth}
        \begin{tikzpicture}[scale=.5]
            \foreach \x in {0,...,10}
            \foreach \y in {0,...,10}
            {
            \node[mystyle1](\x-\y) at (\x,\y){};
            }
            \node[mystyle2](0-0) at (0,0){};
            \node[mystyle2](1-3) at (1,3){};
            \node[mystyle2](2-6) at (2,6){};
            \node[mystyle2](2-1) at (2,1){};
            \node[mystyle2](4-2) at (4,2){};
            \node[mystyle2](3-4) at (3,4){};
            \node[mystyle2](4-7) at (4,7){};
            \node[mystyle2](5-5) at (5,5){};
            \node[mystyle2](6-8) at (6,8){};
            
            \node[mystyle3](0-1) at (0,1){};
                        \node[mystyle3](0-2) at (0,2){};
            \node[mystyle3](0-3) at (0,3){};
                        \node[mystyle3](0-4) at (0,4){};
                        
            \node[mystyle4](3-9) at (3,9){};
            \node[mystyle4](6-3) at (6,3){};

        \end{tikzpicture}
    \end{minipage}%
    \caption{}{Red: elements of S \\ Blue: representatives of other cosets \\ Green: two starting points of copies of S, extending coverage of the red coset}
\end{figure}

Consider $\mathbb{Z}^2$ and maps
\[
\phi_1(x, y) = 3x-y
\]
\[
\phi_2(x,y) = x-2y
\]
The kernels are respectively $\langle e_1 + 3 e_2 \rangle$ and $\langle 2 e_1 + e_2 \rangle$.

\vspace*{.5cm}
Figure A.1 depicts the tile shape $S$ that is produced by Algorithm \ref{rankone} when $M=6$.  It also shows a representative of the other 4 cosets of $\mathbb{Z}^2 / \langle e_1 + 3 e_2, 2e_2 + e_3 \rangle$.  That is, it depicts $T_3$ of Algorithm \ref{tiling}.  $T_2$ is $\{0\}$ for this example, because $\langle e_1 + 3 e_2, 2e_2 + e_3 \rangle$ has the same rank as $\mathbb{Z}^2$.  The two green dots correspond to two of the smallest elements of $T_1$, which accounts for the finite size of $M$.  If we were tiling, copies of $S$ would be translated to start there.

\vspace*{.5cm}
We also alluded to a situation in which infinite data re-use is possible, when the HBL primal LP is infeasible.  Consider the computation lattice is $\mathbb{Z}^2$ with a single map $\phi_1(x,y) = x$.  The entire tile $\langle e_2 \rangle$ is mapped to $0$.  So with $M=1$, we could perform an infinite number of calculations.

\subsection{Multiple Tilings and One with non-zero $T_2$}
Consider the following loop nest:

\vspace*{.5cm}
Loop over $e_1$,$e_2$,$e_3$,$e_4$ \\ 
\-\hspace{1cm} inner loop($A_1[e_1,e_3], \, A_2[e_2,e_4], \, A_3[e_1,e_2,e_3+e_4], \, A_4[e_1+e_2, e_3,e_4]$)

\vspace*{1cm}
This corresponds to the HBL problem in a $d=4$ dimensional lattice, and linear maps
\[ \phi_1 = \left( \begin{array}{cccc}
1 & 0 & 0 & 0 \\
0 & 0 & 1 & 0 \end{array} \right), \, 
\phi_2 = \left( \begin{array}{cccc}
0 & 1 & 0 & 0 \\
0 & 0 & 0 & 1 \end{array} \right), \,
\phi_3 = \left( \begin{array}{cccc}
1 & 0 & 0 & 0 \\
0 & 1 & 0 & 0 \\
0 & 0 & 1 & 1 \end{array} \right), \, 
\phi_4 = \left(\begin{array}{cccc} \,
1 & 1 & 0 & 0 \\
0 & 0 & 1 & 0 \\
0 & 0 & 0 & 1 \end{array} \right)
 \]
 As the lattice is $\mathbb{Z}^4$, we can efficiently calculate the subgroups needed to formulate the primal.  They are:
 
\vspace*{.25cm} 
 $
 \langle e_3-e_4 \rangle, \, \langle e_1-e_2 \rangle, \, \langle e_1-e_2, e_3-e_4 \rangle, \,\langle e_2, e_4 \rangle , $

 $ \langle e_1, e_3 \rangle, \, \langle e_2,e_3,e_4 \rangle, \, \langle e_1, e_2, e_4 \rangle, \, \langle e_1, e_2, e_4 \rangle, \, \langle e_1, e_2, e_3 \rangle, \, \langle e_1,e_2,e_3,e_4 \rangle
 $
 
 \vspace*{.25cm}
When we solve the primal LP using the computational algebra software Magma, $s = (0, .5, .5, .5)$ so $s_{\text{HBL}} = 1.5$.

\vspace*{.5cm}
 Solving the dual LP with Magma, we get a solution supported on subgroups
 \[
 \langle e_1,e_2,e_3,e_4 \rangle, \, \langle e_1-e_2, e_3-e_4 \rangle
 \]
 with dual values of $.25$ for each.  One flag decomposition (Def. \ref{FP}) of this flag uses subgroups
 \[
 \langle e_1-e_2, e_3-e_4 \rangle, \langle e_1, e_3 \rangle
 \]
 with dual values $.5$ and $.25$ respectively. Inputting this to Algorithm \ref{tiling} in the natural way, the tiling set would be
 \[
 \{a_1\cdot(e_1-e_2)+a_2 \cdot (e_3-e_4) + a_3 \cdot e_1 + a_4 \cdot e_3 \, | \, a_i \in \mathbb{Z}, \, 0 \leq a_1, a_2 \leq \floor{M^{.5}}-1, \, 0 \leq a_3, a_4 \leq \floor{M^{.25}} -1 \}
 \]
 
 \vspace*{.5cm}
This example illustrates that asymptotically optimal tilings can substantially differ. Checking the following by hand, another asymptotically optimal tiling for this problem defines the tiling set to be

\[
\{a_1 \cdot e_1 + a_2 \cdot e_3, + a_3 \cdot(e_2-e_3) \, | \, a_i \in \mathbb{Z}, \, 0 \leq a_i \leq \floor{M^{.5}}-1 \}
\]
This shape is fundamentally different than the previous.  Notably, direction $e_4$ is completely ignored.  Algorithm \ref{tiling} for this example outputs $T_2 = \langle e_4 \rangle$ to account for the fact that the tile is within a rank 3 subgroup of $\mathbb{Z}^4$.  Although in this example there is still an optimal full dimensional tile, the next example uses a rank 7 tile in $\mathbb{Z}^8$ and this is the only optimal tiling we found.  Examples exhibiting this behavior are plentiful.
 
\subsection{Need for Flags}
We would like to exhibit a situation in which the dual LP excludes a dual vector which yields a valid tiling. This motivates our use of flags to find tilings.  Consider the computation lattice $\mathbb{Z}^8$ with maps $\phi_1$, $\phi_2$, $\phi_3$ have kernels given respectively by 
\[ \left( \begin{array}{cccccccc}
1 & 0 & 0 & 0 & 0 & 0 & 0 & 0 \\
0 & 1 & 0 & 0 & 0 & 0 & 0 & 0 \\
0 & 0 & 1 & 0 & 0 & 0 & 0 & 0 \end{array} \right), \, 
\left( \begin{array}{cccccccc}
0 & 1 & 0 & 0 & 0 & 0 & 0 & 1 \\
0 & 0 & 1 & 0 & 1 & 1 & 0 & 1 \\
0 & 0 & 0 & 0 & 0 & 1 & 0 & 1 \\
1 & 1 & 0 & 0 & 0 & 0 & 1 & 0 \end{array} \right), \,
\left( \begin{array}{cccccccc}
1 & 0 & 1 & 0 & 0 & 0 & 0 & 0\\
0 & 1 & 1 & 1 & 0 & 0 & 0 & 0 \\
0 & 0 & 0 & 0 & 0 & 1 & 0 & 0 \\
0 & 0 & 0 & 0 & 0 & 0 & 1 & 1 \end{array} \right)
\]

Since there are only 3 maps, this is another case in which primal LP can be formulated efficiently.  The solution that Magma outputs is is $s = (1,1,0)$ so that $s_{\text{HBL}} = 2$.  Moreover, the output dual solution it outputs is supported on subgroups
\[
\langle e_1 + e_7 - e_8, e_2 + e_8, e_3 + e_7 - e_8, e_5 + e_7 - e_8, e_6 + e_8 \rangle , \, \langle e_1, e_2, e_3, e_6, e_7 + e_8 \rangle
\]

with dual values of $.1$ and $.3$ respectively.  Applying one iteration Algorithm \ref{MD} and forming a flag decomposition of Def. \ref{FP}, the new dual solution is supported on subgroups

\[
Y_1 := \langle e_1 + 2 e_6 + e_7 + e_8, e_3 + 2 e_6 + e_7 + e_8, e_2 -e_6 \rangle, \, Y_2 := \langle e_1, e_2 \rangle, \, Y_3 := \langle e_5, e_7 \rangle
\]
with dual values $.4, .3, .1$.

\vspace*{.5cm}
Now for the main point of the example; we calculate $C_1(y)$ as in Eq. \ref{feasibility1}.  It turns out $C_3(y) = 1.2$.  Indeed,

\[
y_{Y_1} \cdot \text{rank}(\phi_3(Y_1))+ y_{Y_2} \cdot \text{rank}(\phi_3(Y_2)) + y_{Y_3} \cdot \text{rank}(\phi_3(Y_3)) = .4 \cdot 2 + .3 \cdot 2 + .1 \cdot 2 = 1.6
\]

Proposition \ref{NPT} is critical here, because $\phi_3(Y_2+Y_1) = \phi_3(Y_1)$, meaning the $.3 \cdot 2$ term is unnecessary.

\end{appendices}

\end{document}